\documentclass[a4paper,12pt]{scrartcl}

\newif\ifsimple
\simpletrue

\ifsimple
\fi

\title{Differential recursion and\\differentially algebraic functions%
 \thanks{Presented 
  at the Second Conference on
  Computability in Europe (CiE~2006), Swansea, Wales, UK, July~2006.
  Supported in part by Research Fellowship (DC1, 18-11700) 
  of the Japan Society for Promotion of Science
  while the author was at Tokyo Institute of Technology.}}
\author{Akitoshi~\textsc{Kawamura}%
 \\
 \large Department of Computer Science\\[-4pt]
 \large University of Toronto}
\date{}

\deffootnote[1em]{1em}{1em}{\textsuperscript{\thefootnotemark}}

\usepackage
{graphicx}

\usepackage[
bookmarks=true,bookmarksnumbered=true,%
pdftitle={Differential recursion and differentially algebraic functions},%
pdfauthor={Akitoshi Kawamura},%
pdfsubject={},%
pdfkeywords={real recursive functions, analog computation, differentially algebraic functions, transcendentally transcendental functions},
linkcolor=blue,anchorcolor=blue,urlcolor=blue,
]{hyperref}

\usepackage[square,comma,numbers,sort&compress]{natbib}

\usepackage[scaled]{helvet}

\ifsimple\else
\usepackage{times}
\usepackage[subscriptcorrection,slantedGreek,nofontinfo,boldalphabet]{mtpro}
\usepackage[mtpcal]{mtpb}
\usepackage{mathrsfs}
\fi

\usepackage{amsmath}
\usepackage{amssymb}
\usepackage{amsthm}
\newtheorem{theorem}{Theorem}
\newtheorem{lemma}[theorem]{Lemma}
\newtheorem{claim}[theorem]{Claim}
\theoremstyle{definition}
\newtheorem{definition}[theorem]{Definition}

\newtheorem{principle}[theorem]{Principle}
\newtheorem*{openproblem}{Open Problem}

\renewcommand{\vec}{}

\newcommand{\Rset}{\mathbf R}

\newcommand{\Zset}{\mathbf Z}
\newcommand{\Nset}{\mathbf N}

\newcommand{\Analytic}{\mathrm C ^\omega}
\DeclareMathOperator{\Derivatives}{\mathbf D}

\DeclareMathOperator{\JX}{\textsc{jx}}
\DeclareMathOperator{\CM}{\textsc{cm}}
\DeclareMathOperator{\PR}{\textsc{dr}}
\DeclareMathOperator{\CampaPR}{\PR _{\mathrm C}}
\DeclareMathOperator{\MN}{\textsc{mn}}

\newcommand{\D}{\mathrm D}
\DeclareMathOperator{\dom}{dom}
\DeclareMathOperator{\arity}{arity}

\newcommand{\pcolon}{\mathrel{:\mskip3mu\joinrel\subseteq}}
\newcommand{\pto}{\rightarrow}

\newcommand{\Proj}[2]{\mathrm{id} _{#1} ^{#2 \to 1}}
\renewcommand{\d}{\mathrm{d}}
\newcommand{\ap}{\mkern1.5mu}

\def\theenumi{\alph{enumi}}
\def\labelenumi{\textup{(\theenumi)}}

\makeatletter
\def\atmosttextstyle#1{\mathchoice{\textstyle#1}{\textstyle#1}{\scriptstyle#1}{\scriptscriptstyle#1}}
\def\notsobig#1{{
  \def\t@mp{#1}
  \xdef#1{\mathop{\atmosttextstyle{\t@mp}}}
}}
\makeatother
\notsobig\sum
\notsobig\prod
\notsobig\bigcup
\notsobig\bigcap

\begin{document}

\maketitle

\begin{abstract}
  Moore introduced a class of real-valued ``recursive'' functions
  by analogy with Kleene's formulation of the standard recursive functions.  
  While his concise definition
  inspired a new line of research on analog computation, 
  it contains some technical inaccuracies.  
  Focusing on his ``primitive recursive'' functions, 
  we pin down what is problematic and 
  discuss possible attempts to remove the 
  ambiguity regarding 
  the behavior of the differential recursion operator on partial functions.  
  It turns out that in any case the purported relation to 
  differentially algebraic functions, 
  and hence to Shannon's model of analog computation, fails.

\ifsimple\else
\smallskip

\fi
\end{abstract}

\section{Introduction}
 \label{section: introduction}

  There are several different kinds of theoretical models that 
  talk about
  ``computability'' and ``complexity'' of real functions. 
\emph{Computable Analysis}~%
\cite{weihrauch00:_comput_analy} 
  and some other equivalent models 
  use approximation in one way or another 
  to bring real numbers 
  into the framework of
  the standard Computability Theory 
  that deals with discrete data in discrete time. 
  Another well-known model is the 
\emph{Blum--Shub--Smale model}~%
\cite{blum97:_compl_real_comput}
  in which continuous quantities are treated as an entity in themselves but 
  the machine still works with discrete clock ticks. 

  A third approach is 
\emph{analog computation}
  in which not only are the data real-valued, 
  but also the transition 
  takes place in continuous time~%
\cite{orponen97:_survey_of_contin_time_comput_theor}. 
  One of the oldest and the best-studied models of such computation is 
  Shannon's \emph{General Purpose Analog Computer}~%
\cite{shannon41:_mathem_theor_differ_analy}
  that models the \emph{Differential Analyzer}~%
\cite{bush31:_differ_analy}, 
  a computing device 
  built and put to use during the thirties through the fifties. 
  The GPAC, 
  after some refinements~%
\cite{pour-el74:_abstr_comput_relat_gener_purpos, 
      lipshitz87:_differ_algeb_replac_theor_analog_comput, 
      graca04:_some_gener_purpos_analog_comput}, 
  was shown capable of generating
  (in a sense) all and only the 
\emph{differentially algebraic} functions. 
  We will explore this class in 
  Section~\ref{section: differential algebraicity} 
  and show that 
  it can be characterized in many different ways. 

\ifsimple
\looseness=-1 
\fi
  Little is known about how such analog models 
  relate to the standard (digital) computability. 
  Moore~%
\cite{moore96:_recur_theor_reals_contin_comput}
  addressed this question
  for his new function classes
  that also try to express the power of GPAC-like computation. 
  In imitation of 
  Kleene's characterization of the usual recursive functions,
  these classes are defined as the closures under certain operators
  that are supposedly real-number versions of 
  primitive recursion and minimization. 
  He makes the following claims, among others, 
  that relate his classes of \emph{real primitive recursive} 
  and \emph{real recursive} functions
  to analog and digital computation, 
  respectively~%
\cite[Propositions 9 and 13]{moore96:_recur_theor_reals_contin_comput}. 

\begin{claim}
 \label{claim: primitive recursive}
  Real primitive recursive functions are 
  differentially algebraic.%
\footnote{%
  Moore writes $M _0$ for the class of real primitive recursion functions. 
  Claim~\ref{claim: primitive recursive} was later 
  replaced by a similar claim~%
\cite[Proposition~2]{campagnolo00:_iterat_inequal_and_differ_in_analog_comput}
  for a more ``restricted'' class~$\mathcal G$ than $M _0$, 
  but its definition is again unclear. 
}
\end{claim}

\begin{claim}
 \label{claim: recursive}
  Each (partial) recursive function on the nonnegative integers 
  (in the standard sense)
  is a restriction of some real recursive function. 
\end{claim}

  Despite its impact on the subsequent study on
  the classes and their variants~%
\cite{campagnolo01:_comput_compl_real_valued_recur, 
      campagnolo01:_upper_and_lower_bound_contin_time_comput, 
      graca04:_some_gener_purpos_analog_comput, 
      mycka04:_real_recur_funct_and_their_hierar, 
      bournez04:_analog_charac_of_elemen_comput}, 
  his work lacked formality in some ways, 
  as already pointed out~%
\cite{campagnolo00:_iterat_inequal_and_differ_in_analog_comput, 
      graca02:_gener_purpos_analog_comop_recur}. 
  In fact, the definition of the classes suffers from ambiguity. 
  In Section~\ref{section: primitive recursive functions} of this paper, 
  we reformulate Moore's theory 
  up to the real primitive recursive functions
  in a mathematically sound way. 
  With the aid of the preparation in 
  Section~\ref{section: differential algebraicity}, 
  we show that, 
  even though our formulation of the class
  seems the most restrictive possible, 
  Claim~\ref{claim: primitive recursive} fails. 
  In section~\ref{section: recursive functions}, 
  we discuss some issues about classes other than 
  the real primitive recursive functions, 
  including Claim~\ref{claim: recursive}. 

  Throughout the paper, 
  we write $\Nset$, $\Zset$, $\Rset$ for the sets of
  nonnegative integers (including $0$), 
  integers and 
  real numbers, respectively. 

\paragraph{Partial functions}

  In this paper, a 
\emph{function} 
  $f \pcolon \Rset ^m \pto \Rset ^n$
  may be \emph{partial}, as opposed to \emph{total}; 
  that is, 
  the set $\dom f$ of $x \in \Rset ^m$ for which 
  the value $f \ap x \in \Rset ^n$ is defined
  is allowed to be a proper subset of $\Rset ^m$. 
  By the 
\emph{restriction} 
  of $f$ to a set $J \subseteq \Rset ^m$ we mean 
  the function~$g$ with $\dom g = J \cap \dom f$ such that $g \ap x = f \ap x$ 
  for every $x \in \dom g$. 
  When $\dom f$ is open, 
  $f$ is said to be \emph{(real) analytic} 
  if 
  for every $a = (a _0, \dots, a _{m - 1}) \in \dom f$
  there are an open set $
J \subseteq \dom f
  $ containing $a$ 
  and a family $(c _p) _{p \in \Nset ^m}$ of $n$-tuples of real numbers
  such that the sum 
\begin{equation}
 \label{eq: 0702171515}
  \sum _{p = (p _0, \dots, p _{m - 1}) \in \Nset ^m}
    c _p
   \cdot 
    (x _0 - a _0) ^{p _0} \cdots (x _{m - 1} - a _{m - 1}) ^{p _{m - 1}}
\end{equation}
  converges to $f \ap x$ 
  for each $x = (x _0, \dots, x _{m - 1}) \in J$
  (regardless of the ordering of summation). 
  See Krantz and Parks~%
\cite[Chapters 1 and 2]{krantz02:_primer_of_real_analy_funct}
  for well-known properties of analytic functions. 
  When $f$ is analytic, 
  we write $\D ^{(a _0, \dots, a _{m - 1})} f$ 
  (and not $
 \partial ^{a _0 + \dots + a _{m - 1}} f 
/
 \partial t _0 ^{a _0} \cdots \partial t _{m - 1} ^{a _{m - 1}}
  $) for the mixed partial derivative of $f$ 
  of order $a _i$ along the $i$\textsuperscript{th} place 
  (which is known to exist). 

  Moore~%
\cite{moore96:_recur_theor_reals_contin_comput} 
  does not explicitly deal with partial functions. 
  We believe that this is responsible for 
  ambiguous and erroneous statements made in his seminal work
  as well as in some of the subsequent works by other authors. 
  Although there are some situations in mathematical analysis 
  where we can pretend that there are no partial functions 
  (namely, when we are only discussing
  properties defined \emph{locally}, 
  such as continuity or analyticity), 
  this is not the case with 
  the notions we want to discuss here. 
  If, say, the above Claim~\ref{claim: recursive} 
  is to make any nontrivial sense, 
  it is clearly inappropriate to talk about 
  ``real recursiveness at $x$,'' 
  as there is a real function 
  which is simple locally 
  but the restriction of which to $\Nset$ is 
  highly complicated in the recursion-theoretic sense.
  We therefore emphasize that
  partial functions must be dealt with seriously, 
  and devote this paper to 
  accordingly reformulating the theory wherever possible. 

\section{Differentially algebraic functions}
 \label{section: differential algebraicity}

  We show some facts about single- and multi-place 
  differentially algebraic functions. 

\begin{theorem}
 \label{theorem: 0607111001}
  Let $m$, $n$ and $i < m$ be nonnegative integers and 
  $f \pcolon \Rset ^m \pto \Rset ^n$ be an analytic function with open domain. 
  Let 
  \textup{(\ref{enumi: Z global})}, 
  \textup{(\ref{enumi: Z local})}, 
  \textup{(\ref{enumi: Z line})} and
  \textup{(\ref{enumi: Z point})} be 
  the following statements: 
\begin{enumerate}
\def\theenumi{\roman{enumi}}
\def\labelenumi{\textup{(\theenumi)}}
 \item \label{enumi: Z global}
  for any open connected set $J \subseteq \dom f$, 
  there is a $\Zset ^n$-coefficient nonzero polynomial~$P$
  such that 
\begin{equation}
 \label{eq: 0605282120}
 P \ap \bigl(
  f \ap x, 
  \D ^{e _i} f \ap x, 
  \D ^{2 \cdot e _i} f \ap x, 
  \ldots, 
  \D ^{(\arity P - 1) \cdot e _i} f \ap x
 \bigr) 
=
 0
\end{equation}
  for all $x \in J$, 
  where $e _i \in \Nset ^m$ is the vector whose 
  $i$\textsuperscript{th} component is $1$ and others are $0$; 
 \item \label{enumi: Z local}
  for each $x _0 \in \dom f$, 
  there are 
  a $\Zset ^n$-coefficient nonzero polynomial~$P$ and 
  an open set~$J$ containing $x _0$
  such that we have 
\eqref{eq: 0605282120}
  for all $x \in J$; 
 \item \label{enumi: Z line}
  for each $x _0 \in \dom f$, 
  there are 
  a $\Zset ^n$-coefficient nonzero polynomial~$P$ and 
  an open interval $J$ containing 
  the $i$\textsuperscript{th} component of $x _0$ 
  such that we have 
\eqref{eq: 0605282120}
  for all $x$ whose $i$\textsuperscript{th} component is in $J$ and 
  whose other components equal those of $x _0$; 
 \item \label{enumi: Z point}
  for each $x \in \dom f$, 
  there is a $\Zset ^n$-coefficient nonzero polynomial~$P$
  satisfying 
\eqref{eq: 0605282120}. 
\end{enumerate}
  Let 
  \textup{(\ref{enumi: Z global}$_\Rset$)}, 
  \textup{(\ref{enumi: Z local}$_\Rset$)} and 
  \textup{(\ref{enumi: Z line}$_\Rset$)}
  be the statements
  obtained by replacing $\Zset$ by $\Rset$ in 
  \textup{(\ref{enumi: Z global})},
  \textup{(\ref{enumi: Z local})} and 
  \textup{(\ref{enumi: Z line})},
  respectively. 
  Then 
  \textup{(\ref{enumi: Z global})}, 
  \textup{(\ref{enumi: Z local})}, 
  \textup{(\ref{enumi: Z line})}, 
  \textup{(\ref{enumi: Z point})}, 
  \textup{(\ref{enumi: Z global}$_\Rset$)}, 
  \textup{(\ref{enumi: Z local}$_\Rset$)} and 
  \textup{(\ref{enumi: Z line}$_\Rset$)} are equivalent. 
\end{theorem}

\begin{proof}
  The implications
  (\ref{enumi: Z global})~$\Rightarrow$ 
  (\ref{enumi: Z local})~$\Rightarrow$ 
  (\ref{enumi: Z line})~$\Rightarrow$ 
  (\ref{enumi: Z point})
  and 
  (\ref{enumi: Z global})~$\Rightarrow$ 
  (\ref{enumi: Z global}$_\Rset$)~$\Rightarrow$ 
  (\ref{enumi: Z local}$_\Rset$)~$\Rightarrow$ 
  (\ref{enumi: Z line}$_\Rset$) are obvious.  
  It has been known that
  (\ref{enumi: Z line}$_\Rset$) $\Rightarrow$ (\ref{enumi: Z line}),
  see Theorem~\ref{theorem: ritt and gourin}. 
  To see (\ref{enumi: Z point})~$\Rightarrow$ (\ref{enumi: Z global}), 
  consider, 
  for each $\Zset ^n$-coefficient polynomial~$P$, 
  the set~$J _P$ of all $
x \in J
  $ satisfying \eqref{eq: 0605282120}. 
  Since by (\ref{enumi: Z point}) 
  these countably many closed sets~$J _P$ cover 
  the open set $J$, 
  one of them must have nonempty interior 
  by Baire Category Theorem~\ref{theorem: baire}. 
  This $J _P$ must then 
  equal $J$ by the Identity Theorem~\ref{theorem: analytic extension}. 
\end{proof}

\begin{definition}
  Let $m$ and $n$ be nonnegative integers. 
  An analytic function $f \pcolon \Rset ^m \pto \Rset ^n$ is 
\emph{differentially algebraic}%
\footnote{%
  Also termed 
  \emph{algebraic transcendental} or
  \emph{hypotranscendental}. 
  Functions \emph{without} this property is said to be
  \emph{transcendentally transcendental} or \emph{hypertranscendental}. 
}
  if for each $i < m$ 
  it satisfies one (or all)
  of the clauses in Theorem~%
\ref{theorem: 0607111001}. 
\end{definition}

  Note that $f$ need not be the \emph{unique} solution of 
\eqref{eq: 0605282120}. 
  For example, every function (with open domain) 
  that is constant on each connected component of its domain is
  differentially algebraic because of 
  the single set of equations $\D ^{e _i} f \ap x = 0$. 

  The clauses~(\ref{enumi: Z local})--(\ref{enumi: Z point}) show that
  being differentially algebraic is a ``local'' property. 

  When $\dom f$ is connected, (\ref{enumi: Z global}) reduces to
  the following statement: 
\begin{enumerate}
 \item[(\ref{enumi: Z global}$'$)]
  there is a $\Zset ^n$-coefficient nonzero polynomial~$P$
  such that we have \eqref{eq: 0605282120} for all $x \in \dom f$. 
\end{enumerate}
  Hence, the clause~(\ref{enumi: Z line}) shows that, 
  as long as $\dom f$ is connected, 
  our definition is equivalent to that of many authors, 
  including Moore~%
\cite{moore96:_recur_theor_reals_contin_comput}, 
  who first state the definition for $m = 1$ by 
  (\ref{enumi: Z global}$'$) and then extend it to $m > 1$ by saying that
  a function is differentially algebraic when 
  it is so as
  a unary function of each argument 
  when all other arguments are held fixed.%
\footnote{%
  Their definition for the case $m = 1$ is 
  slightly different from (\ref{enumi: Z global}$'$)
  in that it replaces \eqref{eq: 0605282120} by $
 P \ap (
  x, 
  f \ap x, 
  \D ^{e _i} f \ap x, 
  \ldots, \allowbreak
  \D ^{(\arity P - 2) \cdot e _i} f \ap x
 ) 
=
 0
  $. 
  But the proof of 
  (\ref{enumi: 0702121119})~$\Rightarrow$
  (\ref{enumi: 0702121120}) of
  Lemma~\ref{lemma: 0702130051} shows that 
  this difference is superficial. 
}
  We proved Theorem~\ref{theorem: 0607111001} 
  in order to use (\ref{enumi: Z global}) 
  to present 
  a counterexample to Claim~\ref{claim: primitive recursive}
  later. 

  Let us characterize differentially algebraic functions
  in yet another way for the case $n = 1$. 
  For a field $E$, its subfield $F$ and a set $B \subseteq E$, 
  we write $F (B)$, agreeing tacitly on $E$, 
  for the smallest subfield of $E$ that includes $F$ and $B$. 
  We write $\overline F$ for the 
\emph{algebraic closure}
  of $F$, that is, 
  the set of those elements of $E$ that annuls some
  $F$-coefficient unary nonzero polynomial. 

  Let $J \subseteq \Rset ^m$ be an open set and 
  consider
  the ring $\Analytic [J]$ of 
  analytic functions $g \pcolon \Rset ^m \pto \Rset$ 
  with $\dom g = J$. 
  Note that
  $\Rset$ is embedded into this ring by
  regarding each $x \in \Rset$ as the constant function 
  taking the value~$x$. 
  To assert
\eqref{eq: 0605282120}
  for all $x \in J$ 
  is to say that
\begin{equation}
 \label{eq: 0605282121}
 P \ap \bigl(
  f, 
  \D ^{e _i} f, 
  \D ^{2 \cdot e _i} f, 
  \ldots, 
  \D ^{(\arity P - 1) \cdot e _i} f
 \bigr) 
=
 0
\end{equation} 
  in $\Analytic [J]$. 
  If $J$ is connected, 
  $\Analytic [J]$ has 
  a quotient field 
  by the Identity Theorem~\ref{theorem: analytic extension}, 
  so the notation $
\Rset (\Derivatives f)
  $ in the following lemma makes sense. 
  We write $\Derivatives f = \{\, \D ^a f \mid a \in \Nset ^m \,\}$. 

\begin{lemma}
 \label{lemma: 0702130051}
  Let $J \subseteq \Rset ^m$ be 
  open and connected. 
  For $f \in \Analytic [J]$, 
  the following are equivalent: 
\begin{enumerate}
 \item \label{enumi: 0702121119}
  $f$ is differentially algebraic; 
 \item \label{enumi: 0702121120}
  $\Derivatives f \subseteq \Rset (B)$ for 
  some finite set $B \subseteq \Derivatives f$; 
 \item \label{enumi: 0702121122}
  $\Derivatives f \subseteq \overline{\Rset (B)}$ for 
  some finite set $B \subseteq \Rset (\Derivatives f)$.
\end{enumerate}
\end{lemma}

\begin{proof}
  The implication (\ref{enumi: 0702121120})~$\Rightarrow$
  (\ref{enumi: 0702121122}) is trivial. 
  The Transcendence Degree Theorem~\ref{theorem: dimension}
  shows (\ref{enumi: 0702121122})~$\Rightarrow$ 
  (\ref{enumi: 0702121119}). 
  For (\ref{enumi: 0702121119})~$\Rightarrow$
  (\ref{enumi: 0702121120}), 
  assume that for each $i$ we have an $\Rset$-coefficient 
  polynomial~$P _i$ with
\begin{equation}
 \label{eq: 0605282359}
  P _i \ap (
   f, 
   \D ^{e _i} f, 
   \D ^{2 \cdot e _i} f, 
   \ldots, 
   \D ^{N _i \cdot e _i} f
  ) 
 =
  0, 
\end{equation}
  where $N _i = \arity P _i - 1$. 
  By choosing $N _i$ to be smallest and 
  then the degree of $P _i$ in the last place to be smallest, 
  we may assume that 
\begin{equation}
  \varXi
 = 
  (\D ^{(0, \dots, 0, 1)} P _i) \ap (
   f, 
   \D ^{e _i} f, 
   \D ^{2 \cdot e _i} f, 
   \ldots, 
   \D ^{N _i \cdot e _i} f
  )
\end{equation}
  is nonzero. 
  Consider the order~$\leq$ on $\Nset ^m$ defined by
  setting $a \leq b$ when $a + c = b$ for some $c \in \Nset ^m$. 
  We will show by induction on $a \in \Nset ^m$ that
  $\D ^a f \in \Rset (\{\, \D ^b f \mid b \leq (N _0, \dots, N _{m - 1})\,\})$.
  The case $a \leq (N _0, \dots, N _{m - 1})$ being trivial, 
  assume that
  $a \geq (N _i + 1) \cdot e _i$ for some $i$. 
  Apply $\D ^{a - N _i \cdot e _i}$ to both sides of 
  \eqref{eq: 0605282359} 
  and calculate using chain rules
  to obtain
\begin{equation}
   \varPsi + \varXi \cdot \D ^a f = 0, 
\end{equation}
  where $\varPsi$ can be written as
  a sum of products of 
  several derivatives of $f$ of order $\leq a$ and $\neq a$, 
  which hence enjoy the induction hypothesis. 
\end{proof}

  Apart from purely theoretical interest, 
  the significance of differentially algebraic functions lies in
  their relation to the 
\emph{General Purpose Analog Computer}, 
  an analog computation model introduced by Shannon~%
\cite{shannon41:_mathem_theor_differ_analy} and 
  later refined by Pour-El~%
\cite{pour-el74:_abstr_comput_relat_gener_purpos}. 
  More precisely, 
  if a function $f \pcolon \Rset \pto \Rset$ with nonempty domain 
  is differentially algebraic, 
  then the restriction of $f$ to some nonempty subset of $\dom f$ 
  is GPAC generable~%
\cite[Theorem~4]{pour-el74:_abstr_comput_relat_gener_purpos}; 
  conversely, if a function $f \pcolon \Rset \pto \Rset$ with nonempty domain 
  is GPAC generable, 
  then the restriction of $f$ to some nonempty subset of $\dom f$ 
  is differentially algebraic~%
\cite[Theorem~2]{lipshitz87:_differ_algeb_replac_theor_analog_comput}. 
  Gra\c ca later 
  considered the
\emph{Polynomial GPAC}, 
  a simpler refinement than Pour-El's, 
  and proved analogous results~%
\cite{graca04:_some_gener_purpos_analog_comput}. 

\section{Real primitive recursive functions}
 \label{section: primitive recursive functions}

  The class of real primitive recursive functions is defined~%
\cite{moore96:_recur_theor_reals_contin_comput} 
  as the smallest class containing some basic functions 
  and closed under the operators specified below. 
  Unfortunately, 
  the original definition contains ambiguity, 
  resulting in some inconsistent claims about the class.
  To remedy this, we shall revisit the definitions carefully 
  in Sections \ref{subsection: jx and cm}
  and \ref{subsection: differential recursion}. 
  Section~\ref{subsection: campagnolo} discusses 
  an alternative approach by Campagnolo. 
  Section~\ref{subsection: counterexample} disproves
  Claim~\ref{claim: primitive recursive}. 

\subsection{Two basic operators}
 \label{subsection: jx and cm}

  The real primitive recursive functions 
  are defined through three operators: 
  \emph{juxtaposition}, \emph{composition} and 
  \emph{differential recursion}. 
  The first two are very simple. 

\begin{definition}
 \label{definition: juxtaposition and composition}
  Given 
  $g _0$, \ldots, $g _{n - 1} \pcolon \Rset ^m \pto \Rset$, 
  define their juxtaposition $
\JX {(g _0, \dots, g _{n - 1})} \pcolon \Rset ^m \pto \Rset ^n
  $ by setting $
\JX {(g _0, \dots, g _{n - 1})} \ap \vec x = 
 (g _0 \ap \vec x, \dots, g _{n - 1} \ap \vec x)
  $ whenever $\vec x \in \dom g _0 \cap \dots \cap \dom g _{n - 1}$. 
  Given $f \pcolon \Rset ^m \pto \Rset ^n$ and 
  $g \pcolon \Rset ^l \pto \Rset ^m$, 
  define their composition $\CM {(f, g)} \pcolon \Rset ^l \pto \Rset ^n$ by 
  setting $
\CM {(f, g)} \ap \vec x = f \ap (g \ap \vec x)
  $
  whenever $\vec x \in \dom g$ and $g \ap x \in \dom f$. 
\end{definition}

  We write $f \circ g$ for $
\CM {(f, g)}
  $. 
  As remarked in Section~\ref{section: introduction}, 
  it is important to define precisely what the operators do 
  on partial functions. 
  Note how Definition~\ref{definition: juxtaposition and composition}
  specifies the domain of the functions constructed. 
  If $g \ap x$ is not defined, 
  neither is $(f \circ g) \ap x$, 
  even if, say, $f$ is a constant function defined everywhere. 
  We thus work in the 
  following (informal) general principle. 

\begin{principle}
 \label{principle: call-by-value}
  For the value of an expression to be defined, 
  the value of each of its subexpressions has to be defined. 
\end{principle}

  We remark that this was not explicitly intended by Moore. 
  In fact, he presents an example to the contrary 
  when he claims~%
\cite[Section~6]{moore96:_recur_theor_reals_contin_comput} 
  that 
  the total function $\overline{\mathrm{inv}} \pcolon \Rset \pto \Rset$
  given by
\begin{equation}
 \label{eq: 0607150051}
  \overline{\mathrm{inv}} \ap x = 
\begin{cases}
 0     & \text{if} \ x = 0 \\
 1 / x & \text{if} \ x \neq 0
\end{cases}
\end{equation}
  can be obtained by composing 
  the binary multiplication with
  $\JX {(\mathrm{zero?}, g)}$, where 
  $\mathrm{zero?}$ is from \eqref{eq: zero test}
  and $g$ is 
  the restriction of $\overline{\mathrm{inv}}$ to $\Rset \setminus \{0\}$.
  Some authors point this out~%
\cite[p.~22]{campagnolo01:_comput_compl_real_valued_recur} 
  and criticize it~%
\cite[p.~47]{graca02:_gener_purpos_analog_comop_recur}. 
  Without discussing which definition is more ``natural,''
  we adopt our restrictive Definition~%
\ref{definition: juxtaposition and composition}, 
  simply because it is not clear
  how to formulate a general definition
  that would admit this construction of $\overline{\mathrm{inv}}$. 

  The operators~$\JX$ and $\CM$ preserve
  analyticity~%
\cite[Proposition~2.2.8]{krantz02:_primer_of_real_analy_funct}. 

\begin{theorem}
 \label{theorem: 0702161640}
  The property of being differentially algebraic 
  is preserved by $\JX$ and $\CM$. 
\end{theorem}

\begin{proof}
  This is trivial for $\JX$. 
  For $\CM$, 
  it suffices to show that if $
f \pcolon \Rset ^m \pto \Rset 
  $ and $g _0$, \ldots, $
g _{m - 1} \pcolon \Rset ^l \pto \Rset 
  $ are differentially algebraic,  
  so is $f \circ g$, where $g = \JX {(g _0, \dots, g _{m - 1})}$. 
  We may assume 
  that $\dom g$ and $J = \dom {(f \circ g)}$ are connected. 
  We use the characterization~(\ref{enumi: 0702121120}) of
  Lemma~\ref{lemma: 0702130051}. 
  Calculate each element of $\Derivatives {(f \circ g)}$ 
  by the chain rule to 
  see that it belongs to 
\begin{equation}
 \label{eq: 0702130056}
  \Rset 
   \bigl(
     \{\, d \circ g \mid d \in \Derivatives f \,\}
    \cup 
     \bigcup _{i = 0} ^{m - 1} 
     \{\, q \mathord\upharpoonright _J \mid q \in \Derivatives g _i \,\}
   \bigr), 
\end{equation}
  where $q \mathord\upharpoonright _J$ means the 
  restriction of $q$ to $J$. 
  By the assumption, 
  there are finite subsets $
A \subseteq \Derivatives f
  $ and $
B _i \subseteq \Derivatives g _i
  $ with $
 \Derivatives f \subseteq \Rset (A)
  $ and $
 \Derivatives g _i \subseteq \Rset (B _i) 
  $ for each $i = 0, \dots, m - 1$. 
  This implies that \eqref{eq: 0702130056} 
  stays unchanged by replacing $\Derivatives f$ by $A$ and 
  $\Derivatives g _i$ by $B _i$. 
\end{proof}

\subsection{The differential recursion operator}
 \label{subsection: differential recursion}

\begin{figure}
\begin{center}
 \includegraphics{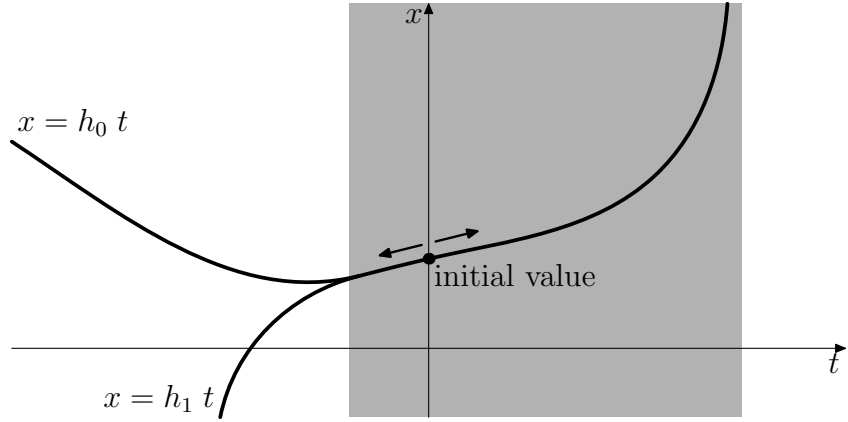}
\end{center}
 \caption{%
  When the equation is satisfied by both $h _0$ and $h _1$ 
  (as well as their restriction to each interval containing the origin), 
  how do we say that 
  the shaded interval is where there is a ``unique solution''?}
 \label{figure: unique}
\end{figure}

  To formulate the third operator, 
  we need a notion of \emph{unique solution} of an integral equation
  of the form \eqref{eq: integral equation} below, 
  where $h$ is the unknown. 
  For example, it sounds natural to say that 
  the tangent function restricted to 
  $(-\pi / 2, \pi / 2)$ uniquely solves $
h \ap t = \int _0 ^t \bigl( 1 + (h \ap \tau) ^2 \bigr) \, \d \tau 
  $.
  But as we are talking about partial functions, 
  the word ``unique'' should be used carefully, 
  because the restriction 
  to any subinterval $J \subseteq (-\pi / 2, \pi / 2)$ containing $0$ also 
  satisfies the equation on $J$. 
  Thus, out of the set~$H$ of all solutions, 
  we need to pick one function that
  deserves to be called the unique solution 
  defined on the largest possible interval 
  (Figure~\ref{figure: unique}). 
  Though Moore did not discuss this, it is not hard 
  to formulate this intuition: 
  for a set $H$ of functions of a type, 
  we say that a function $h \in H$ is 
\emph{unique} 
  in $H$ if 
  the restriction of any function in $H$ to $\dom h$ 
  is a restriction of $h$. 

\begin{definition}
 \label{definition: conservative differential recursion}
  Let $f \pcolon \Rset ^m \pto \Rset ^n$ and 
  $g \pcolon \Rset ^{m + 1 + n} \pto \Rset ^n$. 
  For each $\vec v \in \Rset ^m$, 
  let $H _{\vec v}$ be the set of all functions 
  $h \pcolon \Rset \pto \Rset ^n$ such that
\begin{enumerate}
 \item \label{enumi: first clause}
  $\dom h$ is either the empty set or a 
  possibly unbounded interval containing $0$, 
 \item 
  $\vec v \in \dom f$ if $\dom h$ is nonempty, 
 \item \label{enumi: integrand defined}
  $(\vec v, \tau, h \ap \tau) \in \dom g$
  for each $\tau \in \dom h$, and
 \item
  every $t \in \dom h$ satisfies
\begin{equation}
 \label{eq: integral equation}
 h \ap t =
  f \ap \vec v + 
   \int _0 ^t
    g \ap (\vec v, \tau, h \ap \tau) \, \d \tau. 
\end{equation}
\end{enumerate}
  Let $K _{\vec v}$ be the set of functions unique in $H _{\vec v}$. 
  By Lemma~\ref{lemma: 0604272255} in the appendix, 
  $K _{\vec v}$ has an element~$h _{\vec v}$ of which 
  all functions in $K _{\vec v}$ is a restriction. 
  Define $
\PR {(f, g)} \pcolon \Rset ^{m + 1} \pto \Rset ^n
  $ by $
 \dom {\bigl( \PR {(f, g)} \bigr)}
=
 \{\, (v, t) \in \Rset ^{m + 1} \mid t \in \dom h _{\vec v} \,\}
  $ and $\PR {(f, g)} \ap (\vec v, t) = h _{\vec v} \ap t$. 
\end{definition}

\begin{definition}
 \label{definition: primitive recursive}
  The class of 
\emph{real primitive recursive} 
  functions is the smallest class 
  containing the nullary functions
  $0 ^{0 \to 1}$, $1 ^{0 \to 1}$, $-1 ^{0 \to 1}$ and 
  closed under $\mathord{\JX}$, $\mathord{\CM}$ and $\mathord{\PR}$. 
\end{definition}

\begin{lemma}
 \label{lemma: primitive recursive functions}
  The following functions are real primitive recursive: 
  for each $n \in \Nset$, 
  the $n$-ary 
  constants $0 ^{n \to 1}$, $1 ^{n \to 1}$, $-1 ^{n \to 1}$; 
  for $n \in \Nset$ and $i = 0, \dots, n - 1$, 
  the $n$-ary projection $\Proj i n$ to the $i$\textsuperscript{th} component; 
  binary $\mathrm{add}$ and $\mathrm{mul}$;
  the functions $\mathrm{inv} _+$ (mapping $x > 0$ to $1 / x$), 
  $\mathrm{sqrt _+}$ (mapping $x > 0$ to $\sqrt x$)
  and $\mathrm{ln}$ (natural logarithm)
  defined on $(0, \infty)$;
  the total functions 
  $\mathrm{sin}$, $\mathrm{cos}$ and
  $\mathrm{exp}$; 
  the circle ratio $\pi$ as a nullary function. 
\end{lemma}

\begin{proof}
  The constant $0 ^{n \to 1}$ is built by $
 0 ^{n \to 1} 
= 
 0 ^{0 \to 1} \circ \JX {(\,)}
  $; similarly for $1 ^{n \to 1}$ and $-1 ^{n \to 1}$. 
  Then inductively define $
 \Proj{i}{i + 1} 
= 
 \PR {(0 ^{i \to 1}, 1 ^{i + 2 \to 1})}
  $ and $
 \Proj i {n + 1} 
=
 \PR {(\Proj i n, 0 ^{n + 2 \to 1})}
  $.  Using these, let $
 \mathrm{add}
=
 \PR {(\Proj 0 1, 1 ^{3 \to 1})}
  $ and $
 \mathrm{mul}
=
 \PR {(0 ^{1 \to 1}, \Proj 0 3)}
  $. 
  For $\mathrm{inv} _+$, define
\begin{align}
   f 
 &
  =
   \PR
   {\bigl(
    1 ^{0 \to 1}, 
      \mathrm{mul}
     \circ
      \JX{\bigl( 
       -1 ^{1 \to 1}, \mathrm{mul} \circ \JX {(\Proj 0 1, \Proj 0 1)}
      \bigr)}
     \circ
      \Proj 1 2
   \bigr)},
 \\
   \mathrm{inv} _+ 
 &
  =
   f \circ
   \bigl(
    \mathrm{add} \circ
    \JX {(\Proj 0 1, -1 ^{1 \to 1})}
   \bigr), 
\end{align}
  or, more colloquially, 
\begin{align}
   f \ap t
 &
  =
   1 -
   \int _0 ^t 
    (f \ap \tau) ^2 \, \d \tau,
 &
   \mathrm{inv} _+ \ap t
 &
  =
   f \ap (t - 1).
\end{align}
  Square root is defined analogously by
\begin{align}
 \label{eq: square root}
   f \ap t
 &
  =
    1
   +
    \int _0 ^t 
     \mathrm{inv} _+ \ap (2 \cdot f \ap \tau) \, \d \tau, 
 &
   \mathrm{sqrt} _+ \ap t
 &
  =
   f \ap (t - 1). 
\end{align}
  Logarithm and exponentiation are analogous, 
  using suitable integral equations. 
  For the trigonometric functions, 
  let $
   \mathrm{sin}
  =
   \Proj 0 2 \circ \mathrm{trig}
  $ and $
   \mathrm{cos}
  =
    \Proj 1 2 \circ \mathrm{trig}
  $, where
\begin{equation}
  \mathrm{trig} 
  = 
   \PR 
   {\bigl(
    \JX {(0 ^{0 \to 1}, 1 ^{0 \to 1})}, 
    \JX{\bigl(
      \Proj 2 3, 
      \bigl(
       \mathrm{mul} \circ \JX {(-1 ^{3 \to 1}, \Proj 1 3)}
      \bigr)
    \bigr)}
   \bigr)}, 
\end{equation}
  which is to say, 
\begin{equation}
   \begin{pmatrix}
    \mathrm{sin} \ap t \\ 
    \mathrm{cos} \ap t
   \end{pmatrix}
 =
   \begin{pmatrix}
    0 \\
    1
   \end{pmatrix}
  + 
  \int _0 ^t
   \begin{pmatrix}
    \mathrm{cos} \ap \tau \\
    - \mathrm{sin} \ap \tau
   \end{pmatrix}
  \, \d \tau.
\end{equation}
  The circle ratio is $
 \pi \ap (\,) = 4 \cdot \mathrm{Arctan} \ap 1
  $, with $\mathrm{Arctan}$ defined by a suitable integral equation. 
\end{proof}

  Some authors say
  ``the function $1 / x$ is real primitive recursive'' to mean 
  that $\mathrm{inv} _+$ is. 
  It is not clear how such assertions without specification of domain
  can be justified. 

  The reader may have felt uncomfortable with the 
  unwieldy process of
  Definition~\ref{definition: conservative differential recursion}
  in picking the right solution $h _{\vec v}$ out of $H _{\vec v}$. 
  This can be simplified 
  if we discuss only real primitive recursive functions, 
  because of the following facts that result from the 
\emph{Uniqueness Theorem}
  for initial value problems~%
\cite{lang93:_real_funct_analy}
  and the 
\emph{Cauchy--Kowalewsky Theorem}~%
\cite[Section~2.4]{krantz02:_primer_of_real_analy_funct}. 

\begin{theorem}
 \label{theorem: preservation under pr}
  Let $f$, $g$, $v$, $H _{\vec v}$ and $h _{\vec v}$ be as in 
  Definition~\ref{definition: conservative differential recursion}. 
\begin{enumerate}
 \item \label{enumi: 0702201533}
  If $g$ is an analytic%
\footnote{%
This fact is often stated with a weaker assumption that $g$ be
\emph{Lipschitz continuous}.} 
  function with open domain,
  $H _{\vec v}$ is
  the set of all restrictions of $h _v$. 
 \item \label{enumi: 0702201534}
  If $
f
  $ and $
g
  $ are analytic functions with open domain, 
  so is $\PR {(f, g)}$. 
\end{enumerate}
\end{theorem}

  The fact~(\ref{enumi: 0702201533}) says that 
  a solution of \eqref{eq: integral equation}
  may diverge to infinity at some point but can never ``branch''
  as in Figure~\ref{figure: unique}, 
  provided $g$ is smooth enough. 
  We therefore could have dispensed with 
  Lemma~\ref{lemma: 0604272255} 
  and simply let $h _{\vec v}$ be 
  the (graph) union of $H _v$,
  so far as real primitive recursive functions are concerned, 
  because they are analytic by (\ref{enumi: 0702201534}). 

\subsection{Campagnolo's differential recursion}
 \label{subsection: campagnolo}

  The clauses~(\ref{enumi: first clause})--%
  (\ref{enumi: integrand defined}) of Definition~%
\ref{definition: conservative differential recursion}
  guarantee that
  the integral equation \eqref{eq: integral equation} makes sense
  for all $t \in \dom h$. 
  The clause (\ref{enumi: integrand defined}), however,
  could be slightly relaxed, 
  since a small set of singularities in the integrand does
  not affect the integration. 
  Define $\CampaPR$ by 
  replacing (\ref{enumi: integrand defined}) with
\begin{enumerate}
 \item[(\ref{enumi: integrand defined}$'$)]
  $(\vec v, \tau, h \ap \tau) \in \dom g$ for
  any $\tau \in \dom h \setminus S$, 
  where $S$ is a countable set of isolated points.
\end{enumerate}
  This is due to Campagnolo~%
\cite[Definition~2.4.2]{campagnolo01:_comput_compl_real_valued_recur}, 
  though he does not present 
  a precise specification of the ``unique'' solution
  as we noted in the Section~\ref{subsection: differential recursion}. 
  The choice between (\ref{enumi: integrand defined}) and 
  (\ref{enumi: integrand defined}$'$) is somewhat similar to 
  the discussion regarding 
  Principle~\ref{principle: call-by-value}. 
  The issue is whether 
  $g \ap (\vec v, \tau, h \ap \tau)$, where $\tau \in [0, t]$, 
  is a ``subexpression'' of 
  the right-hand side of the equation~\eqref{eq: integral equation}. 
  Without going into the philosophical discussion to 
  ask which is ``natural,''
  we point out some differences this choice incurs. 

\begin{figure}
\begin{center}
 \includegraphics{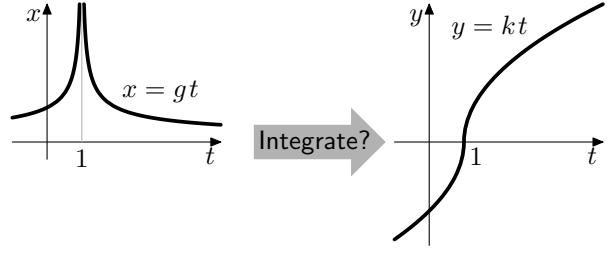}
 \caption{Integrand with a singularity.}
  \label{figure: integrand with a singularity}
\end{center}
\end{figure}

  Theorem~\ref{theorem: preservation under pr}~(\ref{enumi: 0702201534}) fails
  if we replace $\PR$ by $\CampaPR$, 
  as the following example shows 
  (Figure~\ref{figure: integrand with a singularity}). 
  The function $g \pcolon \Rset \pto \Rset$ defined by $
\dom g = \Rset \setminus \{1\}
  $ and $
 g \ap t
=
 \mathrm{inv} _+ \ap \bigl( 
  \mathrm{sqrt} _+ \ap \bigl( \mathrm{sqrt} _+ \ap (t - 1) ^2 \bigr) 
 \bigr)
=
 1 /\! \sqrt{\lvert t - 1 \rvert} 
  $ is real primitive recursive 
  by Lemma~\ref{lemma: primitive recursive functions}. 
  But $
k = \CampaPR {(-2 ^{0 \to 1}, g \circ \Proj 0 2)} 
 \pcolon \Rset \pto \Rset
  $, 
  where $-2 ^{0 \to 1}$ is the constant function with value~$-2$,
  is the total function given by
\begin{equation}
 \label{eq: 06071790810}
 k \ap t = 
  \begin{cases} 
    + 2 \cdot \sqrt{ t - 1} & \text{if $t \geqslant 1$}, \\
    - 2 \cdot \sqrt{-t + 1} & \text{if $t < 1$}, 
  \end{cases}
\end{equation}
  which is not differentiable at $1$. 
  Note that $\PR {(-2 ^{0 \to 1}, g \circ \Proj 0 2)}$ is 
  its restriction to $(-\infty, 1)$
  and thus analytic. 
  For a subtler example, recall the equation
  \eqref{eq: square root} for $\mathrm{sqrt} _+$; 
  with $\CampaPR$, 
  the same equation produces the square root function
  defined on $[0, \infty)$, rather than on $(0, \infty)$. 

  This breaks the assumption of 
  Theorem~\ref{theorem: preservation under pr}~(\ref{enumi: 0702201533})
  and thus gives rise to incomparable functions in $H _{\vec v}$ 
  when, say, $f = 1 ^{0 \to 1}$ and $g = k \circ \Proj 1 2$, 
  with $k$ from \eqref{eq: 06071790810}; 
  that is, the equation
\begin{equation}
 h \ap t = 1 + \int _0 ^t k \ap (h \ap \tau) \, \d \tau
\end{equation}
  has two solutions 
  that take different values at a point. 

  Keeping the class 
  analytic also conforms to
  Moore's intention%
~\cite[Definition~9]{moore96:_recur_theor_reals_contin_comput} 
  to make the equation \eqref{eq: integral equation} 
  equivalent to 
\begin{equation}
 h \ap 0 = f \ap v, \qquad
 \D ^1 h \ap t = g \ap (v, t, h \ap t), 
\end{equation}
  which would not make sense for non-differentiable $h$. 

\subsection{A primitive recursive but not differentially algebraic function}
 \label{subsection: counterexample}

  Claim~\ref{claim: primitive recursive} would not make sense if 
  we adopted $\CampaPR$ in defining real primitive recursive functions,
  because there would then arise non-analytic functions, 
  as we noted above. 
  We now show that, 
  even under our restrictive definition with the 
  analyticity-preserving $\PR$, 
  the claim fails. 

  Define 
  $\Check \varGamma \pcolon \Rset ^2 \pto \Rset$ by 
  $\dom \varGamma = (0, \infty) ^2$ and
\begin{equation}
 \label{eq: 0604151409}
  \Check \varGamma \ap (R, x)
 = 
  \int _{1 / R} ^R
   \exp {\bigl( (x - 1) \cdot \ln t - t \bigr)}
  \, \d t. 
\end{equation}
  Define 
  Euler's \emph{gamma function} $
\varGamma \pcolon \Rset \pto \Rset$ by $\dom \varGamma = (0, \infty)
  $ and
\begin{equation}
 \label{eq: 0604121721}
  \varGamma \ap x = \lim _{R \to \infty} \Check \varGamma \ap (R, x).
\end{equation}
  It can be verified 
  that this value converges and satisfies
\begin{equation}
 \label{eq: 0604151413}
  \D ^n \varGamma \ap x
 =
  \lim _{R \to \infty} \D ^{(0, n)} \Check \varGamma \ap (R, x)
\end{equation}
  for each $n \in \Nset$ and $x \in (0, \infty)$. 
  H\"older showed that $\varGamma$ is 
  not differentially algebraic~%
\cite{hoelder86:_ueber_eigen_gammaf_differ}. 

  We do not know if $\varGamma$ is real primitive recursive, 
  but $\Check \varGamma$ is easily shown real primitive recursive, 
  using Lemma~\ref{lemma: primitive recursive functions}. 
  However, 
  contrary to Claim~\ref{claim: primitive recursive}, 
  it is not differentially algebraic. 
  For assume that it were. 
  We would then have 
  a nonzero polynomial $P$ such that
\begin{equation}
 \label{eq: 0604131901}
  P \ap \bigl( 
   \Check \varGamma \ap (R, x), 
   \D ^{(0, 1)} \Check \varGamma \ap (R, x), 
   \dots, 
   \D ^{(0, \arity P - 1)} \Check \varGamma \ap (R, x) 
  \bigr)
 =
  0
\end{equation}
  for each $(R, x) \in (0, \infty) ^2$. 
  Note that we used the characterization~(\ref{enumi: Z global})
  of Theorem~\ref{theorem: 0607111001}
  in order to take $P$ independent of $R$. 
  We take the limit of \eqref{eq: 0604131901} as 
  $R \to \infty$, which by \eqref{eq: 0604151413} yields
\begin{equation}
  P \ap ( 
   \varGamma \ap x, 
   \D ^1 \varGamma \ap x, 
   \dots, 
   \D ^{\arity P - 1} \varGamma \ap x 
  )
 =
  0, 
\end{equation}
  contradicting H\"older. 


\section{Other classes and related works}
 \label{section: recursive functions}

  This section discusses 
  some other operators 
  introduced by Moore and other authors. 

\subsection{Minimization and Moore's real recursive functions}

  For a function $f \pcolon \Rset ^{m + 1} \pto \Rset$, 
  Moore defines $\MN f \pcolon \Rset ^m \pto \Rset$ by
\begin{equation}
 \label{eq: definition of minimization}
 \MN f \ap \vec v = 
\begin{cases}
  t ^+ = \inf {\{\, t \geq 0 \mid f \ap (\vec v, t) = 0 \,\}} 
 &
  \text{if} \ t ^+ < -t ^-, 
 \\
  t ^-  = \sup {\{\, t \leq 0 \mid f \ap (\vec v, t) = 0 \,\}}
 &
  \text{otherwise}. 
\end{cases}
\end{equation}
  The class of
\emph{real recursive} 
  functions%
\footnote{%
  This ``recursiveness'' of Moore's 
  should not be confused with 
  the same word also used in the context of Computable Analysis. 
  As we see in 
  Appendix~\ref{section: appendix: iteration}, 
  Moore's real recursive functions can even be discontinuous. 
} 
  is the smallest class 
  containing all real primitive recursive functions and 
  closed under 
  $\mathord{\JX}$, $\mathord{\CM}$, $\mathord{\PR}$ and $\mathord{\MN}$. 

  Moore states the definition of $\MN$ 
  in a way that leaves ambiguous 
  whether \eqref{eq: definition of minimization} 
  has a value when, say, $
\dom f = \Rset ^m \times [1, \infty)
  $ and $f \ap (v, t) = 2 - t$ for all $t \geq 1$. 
  Should it have the value~$2$, 
  or be left undefined because 
  ``the zero-searching program gets stuck''? 

  It turns out that, 
  whichever definition we choose, 
  Moore's claim about iteration~%
\cite[Proposition~11]{moore96:_recur_theor_reals_contin_comput} 
  remains true, 
  in the following modified form. 
  Since the original proof again forgets partial functions, 
  we present a new proof in Appendix~\ref{section: appendix: iteration}. 

\begin{lemma}
 \label{lemma: iteration preserves recursiveness}
  If $f \pcolon \Rset ^m \pto \Rset ^m$ is real recursive, 
  there is a real recursive function 
  $g \pcolon \Rset ^{m + 1} \pto \Rset ^m$ 
  that extends the function $g'$ defined by 
  $
 \dom g' 
=
 \bigl\{\, (v, k) \in \Rset \times (\Nset \setminus \{0\}) \bigm| v \in \dom f ^k \,\bigr\}
  $ and
  $g' \ap (\vec v, k) = f ^k \ap \vec v$
  for all $(v, k) \in \dom g'$, 
  where $f ^k = \underbrace{f \circ \dots \circ f} _k$. 
\end{lemma}

  We have to note, however, that 
  the class of real recursive functions is probably not well-behaved, 
  since, 
  with $\MN$ producing non-smooth functions, 
  the class no longer enjoys 
  Theorem~\ref{theorem: preservation under pr}. 
  We therefore doubt the significance of Claim~\ref{claim: recursive}, 
  although it could be justified 
  by using Lemma~\ref{lemma: iteration preserves recursiveness} to
  simulate Turing machines as Moore did. 

\subsection{Linear differential recursion}

  We have seen that 
  many of the problems in Moore's original work
  were caused by
  failure to deal with partial functions properly. 
  Some authors avoid this trouble by
  studying only operators preserving totality, 
  so that partial functions never come into discussion. 
  Campagnolo\- and Moore~%
\cite{campagnolo01:_upper_and_lower_bound_contin_time_comput}
  take this path by considering
  \emph{linear differential recursion}
  in place of $\PR$. 
  For classes defined by this operator, 
  some relationships with digital computation
  are known~%
\cite{campagnolo01:_comput_compl_real_valued_recur, 
      bournez04:_analog_charac_of_elemen_comput}. 

\subsection{Open problems}

  Claim~\ref{claim: primitive recursive} 
  has been the 
  main rationale for calling variants of Moore's classes
  a model of analog computation. 
  Now that we have lost it, 
  an important challenge is the following. 

\begin{openproblem}
  Find a subclass of our real primitive recursive functions, 
  preferably with an equally simple definition, 
  that has a close relationship to the differentially algebraic functions. 
\end{openproblem}

  Another direction would be to reformulate further
  the rest of Moore's work, as well as other authors' works 
  that also suffer from the same kind of ambiguity. 
  For example, 
  it may be interesting to 
  work out Mycka and Costa's class
  arising from the operator of taking limits~%
\cite{mycka04:_real_recur_funct_and_their_hierar}. 

\section*{Acknowledgement}

  The author thanks 
  Ma\-ri\-ko~Ya\-sugi at Kyoto Sangyo University 
  for the discussion that led to this work. 
  Comments by J.\thinspace F.~Costa at Instituto Superior T\'ecnico
  helped improve the presentation of 
  Section~\ref{subsection: differential recursion}. 

\clearpage
\begin{small}
\bibliographystyle{jipsj}

\end{small}

\appendix\small

\section{Old results}

  We list some known theorems that we used in 
  Section~\ref{section: differential algebraicity}. 

  The following
\emph{Baire Category Theorem}
  is used in the proof Theorem~\ref{theorem: 0607111001}. 

\begin{theorem}
 \label{theorem: baire}
  Let $J$ be a subset of\/ $\Rset ^m$. 
  The union of countably many closed subsets of $J$ with empty interior 
  has empty interior. 
\end{theorem}

\begin{proof}
\newcommand{\Ball}{\mathrm{B}}
  Let $J _0$, $J _1$, \ldots be closed subsets of $J$ with empty interior,
  and $U$ be any nonempty open subset of $J$. 
  We will show that $U \setminus \bigcup _{P \in \Nset} J _P$ is nonempty. 
  For each $P \in \Nset$, 
  we take $x _P \in \Rset ^m$ and $\varepsilon _P \in \Rset$ as follows. 
  Write $
 \Ball (x, \varepsilon) 
  $ for the open set of points in $J$
  whose distance from $x$ is less than $\varepsilon$. 
  Let $x _0 \in U$ and $\varepsilon _0 \in (0, 1)$ be 
  such that $\Ball (x _0, \varepsilon _0) \subseteq U$. 
  For each $P \in \Nset$, 
  let $x _{P + 1} \in U$ and $
\varepsilon _{P + 1} \in (0, 2 ^{-P - 1})
  $ be such that $
 \Ball (x _{P + 1}, \varepsilon _{P + 1}) 
\subseteq
 \Ball (x _P, \varepsilon _P) \setminus J _P
  $. 
  This is possible because 
  $\Ball (x _P, \varepsilon _P) \setminus J _P$ is 
  open and nonempty, 
  since $J _P$ is closed and has empty interior. 
  As $P$ tends to infinity, 
  $x _P$ converges to 
  a point in $U \setminus \bigcup _{P \in \Nset} J _P$. 
\end{proof}

  The proof of Theorem~\ref{theorem: 0607111001} also uses the 
  following 
\emph{Identity Theorem}
  (for real analytic functions of several variables), 
  also known as the 
\emph{Principle of Analytic Continuation}. 
  It can be proved by 
  straightforwardly generalizing
  the same assertion for unary functions~%
\cite[Section~1.2]{krantz02:_primer_of_real_analy_funct}. 

\begin{theorem}
 \label{theorem: analytic extension}
  An analytic function with open connected domain 
  that vanishes on an open set vanishes everywhere. 
\end{theorem}

  Let $J \subseteq \Rset$ be an open interval. 
  It is well known that 
  functions $u _0$, \dots, $u _{k - 1} \in \Analytic [J]$ 
  are linearly dependent if and only if
  the determinant $
   \bigl\lvert {(\D ^i u _j) _{i, j = 0, \dots, k - 1}} \bigr\rvert
  $, called their
\emph{Wronskian}, 
  is zero. 
  Using this fact, 
  Ritt and Gourin~%
\cite{ritt27:_assem_theor_proof_exist_trans_trans_funct} 
  showed 
  (\ref{enumi: Z line}$_\Rset$) $\Rightarrow$ (\ref{enumi: Z line}) 
  of Theorem~\ref{theorem: 0607111001}. 

\begin{theorem}
 \label{theorem: ritt and gourin}
  Let $J \subseteq \Rset$ be an open interval and let $f \in \Analytic [J]$.
  If we have
\begin{equation}
 \label{eq: 0702172230}
 P \ap (
  f, 
  \D ^1 f, 
  \D ^2 f, 
  \ldots, 
  \D ^{\arity P - 1} f
 ) 
=
 0 
\end{equation}
  for some $\Rset$-coefficient nonzero polynomial~$P$, 
  then we have 
  \eqref{eq: 0702172230}
  for some $\Zset$-coefficient nonzero polynomial~$P$. 
\end{theorem}

\begin{proof}
  By the assumption, 
  there is a finite set $
B \subseteq \Nset ^{\arity P}
  $ such that the functions
\begin{equation}
 \label{eq: 0605281357}
          f ^{\nu _0} 
    \cdot (\D f) ^{\nu _1} 
    \cdots
    (\D ^{\arity P - 1} f) ^{\nu _{\arity P - 1}}, 
 \qquad
  \text{for} \ (\nu _0, \dots, \nu _{\arity P - 1}) \in B, 
\end{equation}
  are linearly dependent. 
  The Wronskian of \eqref{eq: 0605281357} thus vanishes, 
  which is a $\Zset$-coefficient polynomial 
  in $f$, $\D ^1 f$, \dots, $\D ^{\arity P + \lvert B \rvert - 1} f$. 
  This polynomial is nonzero, since 
  otherwise \eqref{eq: 0605281357} would be 
  linearly dependent for arbitrary $
f
  $, which is absurd. 
\end{proof}

  One direction of Lemma~\ref{lemma: 0702130051} uses the following
\emph{Transcendence Degree Theorem}. 

\begin{theorem}
 \label{theorem: dimension}
  Let $F$ be a subfield of a field~$E$ and 
  $D$ be a subset of $E$. 
  If $D \subseteq \overline{F (B)}$ 
  for some finite set $B \subseteq E$, 
  then 
  $D \subseteq \overline{F (C)}$
  for some finite set $C \subseteq D$. 
\end{theorem}

\begin{proof}
  For each $d \in D$, 
  the assumption gives
\begin{equation}
 \label{eq: 0702122314}
 d ^l = \sum _{j = 0} ^{l - 1} \beta _j \cdot d ^j
\end{equation}
  for some $l \in \Nset \setminus \{0\}$ and 
  $\beta _j \in F (B)$. 
  Suppose that for some $d = d _0 \in D$, 
  this equation contains
  some $b \in B \setminus D$, 
  since otherwise we are done. 
  Then we can rewrite \eqref{eq: 0702122314} as
\begin{equation}
 \label{eq: 0702122313}
 b ^k = \sum _{i = 0} ^{k - 1} \alpha _i \cdot b ^i
\end{equation}
  for some $k \in \Nset \setminus \{0\}$ and 
  $\alpha _i \in F (B')$, 
  where $B' = (B \setminus \{b\}) \cup \{d _0\}$. 

  For each $d \in D$ and $t \in \Nset$, 
  we can substitute \eqref{eq: 0702122314} and \eqref{eq: 0702122313}
  repeatedly in $d ^t$ to write
\begin{equation}
 d ^t = \sum _{i = 0} ^{k - 1} \sum _{j = 0} ^{l - 1} 
  \gamma _{i, j} \cdot c ^i \cdot d ^j
\end{equation}
  for some $\gamma _{i, j} \in F (B')$.
  The $k \cdot l + 1$ elements 
  $1$, $d$, $d ^2$, \ldots, $d ^{k \cdot l}$
  are hence linearly dependent over $F (B')$. 
  We have thus found a set~$B'$ with $D \subseteq \overline{F (B')}$ 
  such that $B' \setminus D$ has strictly less elements than $B \setminus D$. 
  Repeat. 
\end{proof}

\section{Maximal unique function}
 \label{section: appendix: unique}

  This section shows that, 
  from a set~$K$ of functions with a certain property, 
  we can choose a function of which all functions in $K$ is a restriction.
  This was used to justify Definition~%
\ref{definition: conservative differential recursion}
  in the presence of non-analytic functions 
  where 
  Theorem~\ref{theorem: preservation under pr}~(\ref{enumi: 0702201533}) 
  does not apply. 

  We say that a set $I \subseteq \Rset$ is 
\emph{$0$-convex} 
  if it is
  either the empty set or a possibly unbounded interval containing $0$.
  Note that the union of $0$-convex sets is $0$-convex. 

  We say that a set~$K$ of functions from $\Rset$ is 
\emph{consistent} 
  if for any $t \in \Rset$, 
  the set
  $\{\, g \ap t \mid g \in K \,\}$ has at most one element. 
  In this case, the 
\emph{union}
  of $K$ means the unique function $k$ such that 
  $\dom h = \bigcup _{g \in K} \dom g$ and 
  for each $t \in \dom h$, there is some $g \in K$ with $h \ap t = g \ap t$. 

\begin{lemma}
 \label{lemma: 0604272255}
  Let $H$ be a set of functions from $\Rset$ 
  with $0$-convex domain. 
  Then the set $K$ of functions unique in $H$ is consistent. 
  Moreover, if its union belongs to $H$, it belongs to $K$. 
\end{lemma}

\begin{proof}
  For the first claim, suppose otherwise. 
  Then there are functions $k _0$, $k _1 \in K$ and 
  $t \in \dom k _0 \cap \dom k _1$ such that $k _0 \ap t \neq k _1 \ap t$. 
  This contradicts the fact that $k _0$ is unique in $H$. 

  For the second claim, 
  suppose that the union~$k$ of $K$ is not unique in $H$. 
  That is, there are a function $g \in H$ and $t \in \dom k \cap \dom g$ 
  such that $g \ap t \neq k \ap t$. 
  There is $k _0 \in K$ for which $t \in \dom k _0$. 
  We have $k _0 \ap t = k \ap t \neq g \ap t$, 
  contradicting the fact that $k _0$ is unique in $H$. 
\end{proof}

  This lemma can be applied to $H = H _{\vec v}$ in 
  the situation of Definition~%
\ref{definition: conservative differential recursion}, 
  because there the union of any consistent subset of $H _{\vec v}$ 
  belongs to $H _{\vec v}$. 

\section{Iteration}
 \label{section: appendix: iteration}

  As we noted, the definition~%
\eqref{eq: definition of minimization}
  of the operator~$\MN$ is ambiguous, 
  as it contains a subexpression $f \ap (v, t)$ that 
  may be undefined for some $(v, t)$. 
  So when is $\MN f \ap v$ defined? 
  Possible answers include: 
\begin{enumerate}
 \item 
  When $t ^+$ and $t ^-$ are defined. 
 \item \label{enumi: 0610211654}
  When at least either $t ^+$ or $t ^-$ is defined; 
  the condition $t ^+ < -t ^-$ will be used 
  only when both are defined. 
\end{enumerate}
  And when is $t ^+$ (resp.\ $t ^-$) defined? 
  Possible answers include: 
\begin{enumerate} \def\theenumi{\roman{enumi}}
 \item \label{enumi: 0610211748}
  When there is $t \geq 0$ (resp.\ $\leq 0$) such that 
  $f \ap (v, t) = 0$ and
  $(v, \tau) \in \dom f$ for all $\tau \in \Rset$. 
 \item \label{enumi: 0610211655}
  When there is $t \geq 0$ (resp.\ $\leq 0$) such that 
  $f \ap (v, t) = 0$ and 
  $(v, \tau) \in \dom f$ for all $\tau \in [-t, t]$ (resp.\ $[t, -t]$).
 \item \label{enumi: 0610211749}
  When there is $t \geq 0$ (resp.\ $\leq 0$) such that 
  $f \ap (v, t) = 0$ and 
  $(v, \tau) \in \dom f$ for all $\tau \in [0, t]$ (resp.\ $[t, 0]$).
 \item 
  When there is $t \geq 0$ (resp.\ $\leq 0$) such that 
  $f \ap (v, t) = 0$. 
\end{enumerate}
  For (\ref{enumi: 0610211748}), (\ref{enumi: 0610211655}) 
  and (\ref{enumi: 0610211749}), 
  we may also consider adding the phrase
  ``except for some countably many isolated $\tau$''
  (compare (\ref{enumi: integrand defined}$'$) in 
  Section~\ref{subsection: campagnolo}).

  Moore's informal explanation
  by a programming language 
  \cite[Section~7]{moore96:_recur_theor_reals_contin_comput}
  seems to suggest (\ref{enumi: 0610211654}) and 
  (\ref{enumi: 0610211655}). 
  However, 
  without discussing which is the ``right'' definition of $\MN$, 
  we show that, whichever we choose, 
  Lemma~\ref{lemma: iteration preserves recursiveness} holds. 
  The following proof is 
  consistent with any of the above $2 \times 7$ possible definitions. 

\begin{figure}
\begin{center}
 \includegraphics{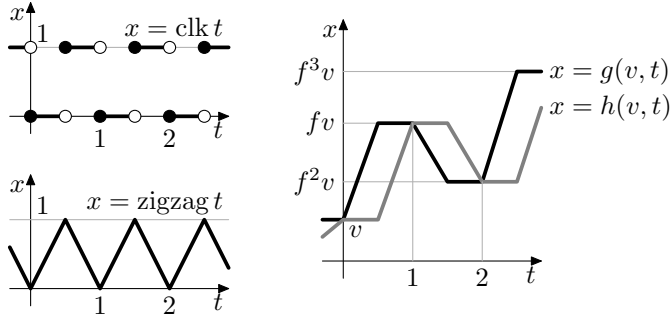}
\end{center}
\caption{Simulating iteration $f ^k \ap \vec v$ by real recursive functions.}
\label{figure: iteration}
\end{figure}

\begin{proof}[Proof of Lemma~\ref{lemma: iteration preserves recursiveness}]
  Denote $\MN f \ap \vec v$ by $\mu t \ldotp f \ap (\vec v, t)$. 
  Let
\begin{align}
 \label{eq: zero test}
  \mathrm{zero?} \ap x & = \mu y \ldotp (x ^2 + y ^2) \cdot (1 - y),
\\
  \mathrm{integer?} \ap x & = \mathrm{zero?} \ap \bigl( \sin {(\pi \cdot x)} \bigr),
 \\
 \label{eq: 0607142300}
  \mathrm{round} \ap x & = x - \mu r \ldotp \mathrm{integer?} \ap (x - r),
\end{align}
  so that \eqref{eq: 0607142300} is
  the unique integer in $(x - 1 / 2, x + 1 / 2]$. 
  We get $\overline{\mathrm{inv}}$
  of \eqref{eq: 0607150051} by 
\begin{equation}
 \overline{\mathrm{inv}} \ap x = \mu t \ldotp x \cdot (x \cdot t - 1). 
\end{equation}
  The above 
  four 
  functions are total. 
  Let 
\begin{equation}
  \mathrm{digit} \ap (x, b, i) = 
   \mathrm{round} \left( \frac x {b ^i} - \frac 1 2 \right) - 
   b \cdot \mathrm{round} \left( \frac x {b ^{i + 1}} - \frac 1 2 \right)
\end{equation}
  for $b > 0$, 
  where $b ^i = \mathrm{exp} \ap (i \cdot \ln b)$. 
  When $b > 1$ and $i$ are integers, 
  $\mathrm{digit} \ap (x, b, i)$ is the digit in 
  $b ^i$'s place when $x$ is written in base-$b$ notation. 
  Define 
\begin{gather}
   \mathrm{clk} \ap t
  = 
   \mathrm{digit} \ap (t, 2, -1),
 \qquad
   \mathrm{zigzag} \ap t 
  = 
   0 + \int _0 ^t (2 - 4 \cdot \mathrm{clk} \ap \tau) \, \d \tau,
\\
 \label{eq: 0612132308}
\begin{pmatrix}
  g \ap (\vec v, t)
 \\
  h \ap (\vec v, t)
\end{pmatrix} 
 =
  \begin{pmatrix} \vec v \\ \vec v \end{pmatrix}
 +
  \int _0 ^t
\begin{pmatrix}
  2 \cdot 
  (1 - \mathrm{clk} \ap \tau) \cdot
   \bigl( 
     f \ap \bigl( h (v, \tau) - \mathrm{clk} \ap \tau \cdot (h \ap (v, \tau) - v) \bigr)
    -
     h \ap (\vec v, \tau) 
   \bigr) 
 \\
     2 
    \cdot 
     \mathrm{clk} \ap \tau
    \cdot 
     \bigl( h \ap (\vec v, \tau) - g \ap (\vec v, \tau) \bigr) 
    \cdot 
     \overline{\mathrm{inv}} \ap (\mathrm{zigzag} \ap \tau) 
\end{pmatrix}
  \, \d \tau, 
\end{gather}
  as in Figure~\ref{figure: iteration}.
  We have $f ^k \ap \vec v = g \ap (\vec v, k - 1 / 2)$ 
  for $k \in \Nset \setminus \{0\}$. 
\end{proof}

  Note that $
\mathrm{clk} \ap \tau \cdot (h \ap (v, \tau) - v) 
  $ in \eqref{eq: 0612132308}
  cannot be dropped, 
  because of Principle~\ref{principle: call-by-value}. 

\end{document}